\documentclass[a4paper,twoside,11pt]{article}
\usepackage[utf8]{inputenc}
\usepackage{graphicx,epstopdf}
\usepackage{amsmath}
\usepackage{amsthm}
\usepackage{amssymb}
\usepackage{color}
\usepackage{mathtools}
\usepackage{subfig}
\usepackage{siunitx}
\usepackage[font=footnotesize,labelfont=bf]{caption}
\usepackage{float}
\usepackage{geometry,calc}
\usepackage{listings}
\usepackage{algorithm}
\usepackage{algpseudocode}
\usepackage{hyperref}
\usepackage{cleveref}

\algrenewcommand\algorithmicrequire{\textbf{Input:}}
\algrenewcommand\algorithmicensure{\textbf{Output:}}

\newtheorem{definition}{Definition}

\newtheorem{theorem}{Theorem}[section]

\newtheorem{lemma}[theorem]{Lemma}
\theoremstyle{remark}

\usepackage[most]{tcolorbox}
\definecolor{darkgreen}{rgb}{0.0, 0.55, 0.0}

\title{New Insights into Population Dynamics from the Continuous McKendrick Model}
\author{Dragos-Patru Covei\thanks{Department of Applied Mathematics, The Bucharest University of Economic Studies, Piata Romana, No. 6, Bucharest, District 1, 010374, Romania. E-mail: dragos.covei@csie.ase.ro}}

\date{}

\begin{document}
\maketitle
	
\begin{abstract}
\noindent
This article presents a comprehensive study of the continuous McKendrick
model, which serves as a foundational framework in population dynamics and epidemiology. The
model is formulated through partial differential equations that describe the
temporal evolution of the age distribution of a population using continuously defined
birth and death rates. In this work, we provide rigorous derivations of the renewal equation,
establish the appropriate boundary conditions, and perform a detailed analysis of the
survival functions. The central result demonstrates that the population approaches extinction if and only if the net
reproduction number $R_{n}$ is strictly less than unity. We present two independent proofs: one based on Laplace transform techniques and Tauberian theorems, and another employing a reformulation as a system of ordinary differential equations with eigenvalue analysis. Additionally, we establish the connection
between the deterministic framework and stochastic process formulations, showing that the McKendrick equation emerges as the fluid limit of an individual-based stochastic model.
		
\end{abstract}


\section{Introduction}

The continuous McKendrick model \cite{mckendrick1925} constitutes a
fundamental framework for studying age-structured populations in
epidemiology and mathematical biology \cite{Brauer2011,Cushing1981,Feller1941,Gurtin1979,Li1988}. In its classical
formulation, the model is expressed as the following initial-boundary value problem:
\begin{equation}
\left\{ 
\begin{array}{lll}
\displaystyle\frac{\partial \rho }{\partial a}(a,t)+\frac{\partial \rho }{\partial t}(a,t)=-\mu \rho (a,t), & \text{for} & 
a,t>0, \\ [2ex]
\displaystyle\rho (0,t)=\int_{0}^{\infty }\beta (a) \rho (a,t)\, da, & \text{for} & t>0, \\ [2ex]
\rho (a,0)=\psi (a), & \text{for} & a\geq 0.
\end{array}%
\right.  \label{mk}
\end{equation}%
Here, the function $\rho(a,t)$ represents the population density of individuals of age $a$ at time $t$. The parameters $\beta(a)$ and $\mu$ denote the age-dependent birth rate and the constant mortality rate, respectively. The function $\psi(a)$ specifies the initial age distribution at time $t=0$. The first equation in \eqref{mk} is the transport equation describing aging with mortality; the boundary condition at $a=0$ states that the influx of newborns equals the total birth rate across all ages; and the initial condition prescribes the starting population structure.

This formulation leads naturally to a renewal
integral equation whose analysis uncovers key insights into the long-term
dynamics of the system. In the present work, we revisit the deterministic
McKendrick model \eqref{mk} and provide complete derivations of its main
results. Our approach employs the Laplace transform to solve the renewal
equation and to derive necessary and sufficient conditions for extinction.
In addition, we complement this analysis with an equivalent formulation in
terms of a system of ordinary differential equations, extending the techniques
employed in earlier studies \cite{Covei2024}. These methods not only
reinforce the classical results but also provide a versatile analytical framework suitable for further
extensions.

Beyond the deterministic setting, we incorporate stochastic effects into the
model by considering an age-structured stochastic process. Under a suitable
scaling, we demonstrate that the stochastic formulation converges to the
deterministic McKendrick model as the population size tends to infinity.
Our stochastic analysis builds on the foundational work by Ethier and Kurtz 
\cite{EthierKurtz} and is further refined in recent studies \cite{Allen2020,Ponosov2020}.

The principal contributions of this work are threefold:

\begin{itemize}
\item[(i)] A rigorous proof establishing that the population tends to extinction if and only if 
$R_{n}<1$, with complete derivations based on the renewal equation and eigenvalue analysis.

\item[(ii)] An explicit demonstration of how the Laplace transform and systems of
ordinary differential equations can be employed to analyze the asymptotic
behavior of the model, providing two independent verification methods.

\item[(iii)] A detailed exposition of the connection between the deterministic and
stochastic formulations, highlighting the impact of demographic noise on population
persistence and the emergence of the McKendrick equation as a fluid limit.
\end{itemize}

In recent years, the integration of individual-based epidemic models with
age-structured partial differential equations has provided new insights into
complex disease dynamics. In particular, the work of Foutel-Rodier et al. 
\cite{FoutelRodier2022PDE} demonstrates how recording the infection age of
individuals yields a low-dimensional yet fully informative representation of
epidemic spread. Building on these ideas, our approach further refines the
connection between microscopic stochastic processes and macroscopic epidemic
behavior. As shown in \cite{FoutelRodier2022PDE}, the McKendrick--von
Foerster framework not only reduces the complexity of high-dimensional ODE
systems but also facilitates reliable statistical inference on key
epidemiological parameters. Inspired by this paradigm, our study aims to
extend these insights to more heterogeneous populations and varied disease
settings.

Our results also extend the findings of Ripoll and Font \cite{RipollandFont2023} by providing a rigorous extinction criterion along with
a thorough stochastic analysis for age-structured populations. These new
insights can be applied within Lotka--Volterra predator-prey frameworks to
refine stability assessments and long-term population predictions.

The remainder of the paper is organized as follows. Section~\ref{sec:det} focuses
entirely on the deterministic McKendrick model, presenting the theoretical framework,
specific cases, and numerical applications. Section~\ref{sec:concept} discusses the conceptual significance of our results, demonstrating that numerical simulations and analytic theory can be unified through the exponential--polynomial structure of $\beta(a)$. Section~\ref{sec:stoch} is
devoted to the stochastic extension of the model and the analysis of
diffusion approximations, including numerical simulations. Section~\ref{sec:disc} discusses the implications of the results. Section~\ref{sec:conc}
concludes the paper. Detailed numerical algorithms with pseudocode and Python implementations are provided in the Appendix.

\section{The Deterministic Continuous McKendrick Model}\label{sec:det}

\subsection{Problem Statement and Main Results}

We consider the continuous McKendrick model \eqref{mk} governing the evolution of an
age-structured population, where $\rho(a,t)$ denotes the density
of individuals of age $a$ at time $t$, and $\mu > 0$ is the constant death rate. The
age-dependent birth rate is assumed to be of the form
\begin{equation}
\beta(a)=e^{-\mu_{1}a}\sum_{i=0}^{n}c_{i}\,a^{i},\quad \text{with }
c_{0}\neq 0,\; c_{i}\geq 0\;(i\geq 1),\text{ and }\mu_{1}>0.  \label{bl}
\end{equation}
This exponential-polynomial structure is biologically motivated: the exponential decay $e^{-\mu_1 a}$ models the decline in fertility with age, while the polynomial captures age-dependent variations in reproductive capacity. The initial condition $\psi(a)$ is taken to be the Dirac delta function centered at $a=0$:
\begin{equation}
\psi(a) = \delta(a), \quad \text{so that} \quad \int_{0}^{\infty}\psi(a)\,da = 1,
\label{eq:initial_dirac}
\end{equation}
representing a population initially concentrated at age zero (all individuals are newborns).

For our analysis, we introduce the \textit{survival function}
\begin{equation}
\pi(a) = e^{-\mu a},
\label{eq:survival}
\end{equation}
which represents the probability that an individual survives from birth to age $a$ under the constant mortality rate $\mu$. The \textit{net reproduction number} is then defined as
\begin{equation}
R_{n} = \int_{0}^{\infty}\beta(a)\,\pi(a)\,da = \sum_{i=0}^{n}c_{i}\frac{i!}{\bar{\mu}^{i+1}}, \quad \text{where } \bar{\mu} = \mu + \mu_{1}.
\label{eq:Rn_def}
\end{equation}
This quantity represents the expected number of offspring produced by an individual over its lifetime and serves as the fundamental threshold parameter for population dynamics.

By setting 
\begin{equation}
B(t) = \rho(0,t),
\label{eq:B_def}
\end{equation}
which represents the birth rate (influx of newborns) at time $t$, one may derive the \textit{renewal equation}
\begin{equation}
B(t) = \varphi(t) + \int_{0}^{t}\beta(a)\,\pi(a)\,B(t-a)\,da,  \label{vol}
\end{equation}
where the function $\varphi(t)$, depending on the initial condition $\psi(a)$, represents the contribution from individuals already present at time $t=0$.

\begin{lemma}[Solution representation via characteristics]\label{lem:characteristics}
The solution $\rho(a,t)$ of the McKendrick system \eqref{mk} can be represented as
\begin{equation}
\rho(a,t) = \begin{cases}
B(t-a)\,\pi(a), & \text{for } t \geq a, \\[1ex]
\displaystyle\psi(a-t)\,\frac{\pi(a)}{\pi(a-t)}, & \text{for } t < a,
\end{cases}  \label{so}
\end{equation}
and the function $\varphi(t)$ in the renewal equation \eqref{vol} is given by
\begin{equation}
\varphi(t) = \int_{t}^{\infty}\beta(a)\,\psi(a-t)\,\pi(a)\,da.
\label{eq:varphi_def}
\end{equation}
\end{lemma}

\begin{proof}
The first equation in \eqref{mk} is a first-order linear transport equation. Using the method of characteristics, we parameterize along characteristic curves defined by $\frac{da}{ds} = 1$ and $\frac{dt}{ds} = 1$, giving $a - t = \text{const}$ along characteristics.

\textbf{Case 1:} For $t \geq a$, the characteristic through $(a,t)$ intersects the boundary $a = 0$ at time $t - a > 0$. Along this characteristic, we have
\[
\frac{d}{ds}\rho(a(s), t(s)) = -\mu\,\rho(a(s), t(s)),
\]
which integrates to $\rho(a,t) = \rho(0, t-a)\,e^{-\mu a} = B(t-a)\,\pi(a)$.

\textbf{Case 2:} For $t < a$, the characteristic intersects the initial line $t = 0$ at age $a - t > 0$. Integrating along the characteristic from $(a-t, 0)$ to $(a, t)$ yields
\[
\rho(a,t) = \psi(a-t)\,e^{-\mu t} = \psi(a-t)\,\frac{\pi(a)}{\pi(a-t)}.
\]

To derive \eqref{eq:varphi_def}, we substitute \eqref{so} into the boundary condition of \eqref{mk}:
\begin{align*}
B(t) &= \int_{0}^{\infty}\beta(a)\,\rho(a,t)\,da \\
&= \int_{0}^{t}\beta(a)\,B(t-a)\,\pi(a)\,da + \int_{t}^{\infty}\beta(a)\,\psi(a-t)\,\frac{\pi(a)}{\pi(a-t)}\,da.
\end{align*}
For our choice $\psi(a) = \delta(a)$ and $\pi(a) = e^{-\mu a}$, the second integral simplifies to
\[
\varphi(t) = \int_{t}^{\infty}\beta(a)\,\delta(a-t)\,e^{-\mu t}\,da = \beta(t)\,e^{-\mu t} = \beta(t)\,\pi(t).
\]
This completes the derivation of the renewal equation \eqref{vol}.
\end{proof}

The main theorem of this work is as follows.

\begin{theorem}[Extinction criterion]\label{main}
The population modeled by the system \eqref{mk} goes to
extinction, i.e.,
\begin{equation}
P(t) = \int_{0}^{\infty}\rho(a,t)\,da \rightarrow 0 \quad \text{as } t \rightarrow \infty,
\label{eq:extinction}
\end{equation}
if and only if the net reproduction number satisfies $R_{n} < 1$.
\end{theorem}

\begin{proof}[Proof of Theorem~\ref{main}]
We present two independent proofs: one based on the Laplace transform of the renewal
equation and an alternative approach via reformulation as a system of ordinary differential equations.

\textbf{Proof I: Laplace Transform Method.}

Taking the Laplace transform (denoted by $\widehat{\cdot}$) of the renewal equation \eqref{vol} and using the convolution theorem yields
\begin{equation}
\widehat{B}(\lambda) = \widehat{\varphi}(\lambda) + \widehat{K}(\lambda)\,\widehat{B}(\lambda),  \label{hat}
\end{equation}
where the kernel transform is
\begin{equation}
\widehat{K}(\lambda) = \int_{0}^{\infty}\beta(a)\,\pi(a)\,e^{-\lambda a}\,da = \sum_{i=0}^{n}c_{i}\frac{i!}{(\lambda + \bar{\mu})^{i+1}}.
\label{eq:kernel_transform}
\end{equation}
For the initial condition $\psi(a) = \delta(a)$, we have $\varphi(t) = \beta(t)\,\pi(t)$, and thus
\begin{equation}
\widehat{\varphi}(\lambda) = \int_{0}^{\infty}\beta(t)\,\pi(t)\,e^{-\lambda t}\,dt = \widehat{K}(\lambda).
\label{eq:varphi_transform}
\end{equation}

Rearranging \eqref{hat} yields
\begin{equation}
\widehat{B}(\lambda) = \frac{\widehat{\varphi}(\lambda)}{1 - \widehat{K}(\lambda)} = \frac{\widehat{K}(\lambda)}{1 - \widehat{K}(\lambda)}.
\label{lap}
\end{equation}

\textit{Step 1: Properties of $\widehat{K}(\lambda)$.} We observe that:
\begin{itemize}
\item[(a)] $\widehat{K}(0) = R_{n} > 0$ by definition \eqref{eq:Rn_def}.
\item[(b)] $\widehat{K}(\lambda)$ is strictly decreasing in $\lambda$ for $\lambda > -\bar{\mu}$, since
\[
\frac{d\widehat{K}}{d\lambda} = -\int_{0}^{\infty}a\,\beta(a)\,\pi(a)\,e^{-\lambda a}\,da < 0.
\]
\item[(c)] $\lim_{\lambda \to \infty}\widehat{K}(\lambda) = 0$.
\item[(d)] $\lim_{\lambda \to -\bar{\mu}^+}\widehat{K}(\lambda) = +\infty$.
\end{itemize}

\textit{Step 2: Existence and location of the dominant root.} The equation $\widehat{K}(\lambda) = 1$, equivalently
\begin{equation}
(\lambda + \bar{\mu})^{n+1} - \sum_{i=0}^{n}c_{i}\cdot i!\cdot(\lambda + \bar{\mu})^{n-i} = 0,
\label{eq:characteristic}
\end{equation}
has a unique real root $\lambda_0$ by the intermediate value theorem and monotonicity of $\widehat{K}$.

\textit{Step 3: Location of $\lambda_0$ relative to zero.}
\begin{itemize}
\item If $R_{n} < 1$, then $\widehat{K}(0) = R_{n} < 1$, so $1 - \widehat{K}(0) > 0$. By continuity and monotonicity, $\lambda_0 < 0$.
\item If $R_{n} = 1$, then $\widehat{K}(0) = 1$, so $\lambda_0 = 0$.
\item If $R_{n} > 1$, then $\widehat{K}(0) > 1$, so $\lambda_0 > 0$.
\end{itemize}

\textit{Step 4: Asymptotic behavior via Tauberian theorems.} 
When $R_{n} < 1$, the function $\widehat{B}(\lambda)$ is analytic in a half-plane $\{\text{Re}(\lambda) > \lambda_0\}$ containing $\lambda = 0$. By classical Tauberian theorems for Laplace transforms (see Feller \cite{Feller1941}), the asymptotic behavior of $B(t)$ as $t \to \infty$ is governed by the singularity of $\widehat{B}(\lambda)$ closest to the imaginary axis, which is the simple pole at $\lambda = \lambda_0$. Specifically,
\begin{equation}
B(t) \sim C\,e^{\lambda_0 t} \quad \text{as } t \to \infty,
\label{eq:B_asymptotic}
\end{equation}
where $C > 0$ is a constant determined by the residue at $\lambda_0$. Since $\lambda_0 < 0$ when $R_{n} < 1$, we conclude $B(t) \to 0$ exponentially.

\textit{Step 5: Extinction of total population.}
Using the representation \eqref{so}, the total population is
\begin{align}
P(t) &= \int_{0}^{\infty}\rho(a,t)\,da = \int_{0}^{t}B(t-a)\,\pi(a)\,da + \int_{t}^{\infty}\psi(a-t)\,\frac{\pi(a)}{\pi(a-t)}\,da \nonumber\\
&= \int_{0}^{t}B(s)\,\pi(t-s)\,ds + e^{-\mu t},
\label{eq:P_integral}
\end{align}
where the last equality uses the substitution $s = t - a$ and $\psi = \delta$. Since $B(t) \to 0$ and $\pi(a) = e^{-\mu a} \to 0$ as their arguments tend to infinity, both terms vanish, yielding $P(t) \to 0$.

Conversely, if $R_{n} \geq 1$, then $\lambda_0 \geq 0$, and $B(t)$ does not decay to zero, ensuring population persistence.

\textbf{Proof II: ODE System and Eigenvalue Analysis.}

When the birth law has the form \eqref{bl}, we can recast the renewal
equation \eqref{vol} into a system of $n+1$ ordinary differential equations.
Define, for $i = 0, 1, \ldots, n$, the auxiliary functions
\begin{equation}
B_{i}(t) = \int_{0}^{t}a^{i}\,e^{-\bar{\mu} a}\,B(t-a)\,da = \int_{0}^{t}(t-s)^{i}\,e^{-\bar{\mu}(t-s)}\,B(s)\,ds.
\label{deriv}
\end{equation}

\textit{Step 1: Derivation of the ODE system.}
Differentiating $B_0(t)$ using Leibniz's rule:
\begin{align*}
B_0'(t) &= \frac{d}{dt}\int_{0}^{t}e^{-\bar{\mu}(t-s)}B(s)\,ds \\
&= e^{-\bar{\mu}\cdot 0}B(t) + \int_{0}^{t}\frac{\partial}{\partial t}\left[e^{-\bar{\mu}(t-s)}\right]B(s)\,ds \\
&= B(t) - \bar{\mu}\int_{0}^{t}e^{-\bar{\mu}(t-s)}B(s)\,ds = B(t) - \bar{\mu}B_0(t).
\end{align*}
From the renewal equation \eqref{vol} with $\varphi(t) = \beta(t)\pi(t)$:
\[
B(t) = \sum_{i=0}^{n}c_i\,t^i\,e^{-\bar{\mu} t} + \sum_{i=0}^{n}c_i\,B_i(t).
\]
For $i \geq 1$, differentiating $B_i(t)$ similarly yields:
\begin{align*}
B_i'(t) &= 0^i\cdot e^0 \cdot B(t) + \int_{0}^{t}\frac{\partial}{\partial t}\left[(t-s)^i e^{-\bar{\mu}(t-s)}\right]B(s)\,ds \\
&= \int_{0}^{t}\left[i(t-s)^{i-1} - \bar{\mu}(t-s)^i\right]e^{-\bar{\mu}(t-s)}B(s)\,ds \\
&= i\,B_{i-1}(t) - \bar{\mu}\,B_i(t).
\end{align*}

Thus, the vector $\mathbf{B}(t) = (B_{0}(t), B_{1}(t), \ldots, B_{n}(t))^{T}$ satisfies
\begin{equation}
\mathbf{B}'(t) = A\,\mathbf{B}(t) + \mathbf{d}(t),  \label{ode}
\end{equation}
where the matrix $A$ has the structure
\begin{equation}
A = \begin{pmatrix}
c_{0} - \bar{\mu} & c_{1} & c_2 & \cdots & c_{n-1} & c_{n} \\
1 & -\bar{\mu} & 0 & \cdots & 0 & 0 \\
0 & 2 & -\bar{\mu} & \cdots & 0 & 0 \\
\vdots & \vdots & \ddots & \ddots & \vdots & \vdots \\
0 & 0 & 0 & \cdots & n & -\bar{\mu}
\end{pmatrix},  \label{ma}
\end{equation}
and the forcing term is $\mathbf{d}(t) = (\varphi(t), 0, \ldots, 0)^{T}$.

\textit{Step 2: Solution of the ODE system.}
With initial condition $\mathbf{B}(0) = \mathbf{0}$, the unique solution is given by the variation of constants formula:
\begin{equation}
\mathbf{B}(t) = \int_{0}^{t}e^{A(t-s)}\,\mathbf{d}(s)\,ds.  \label{bc}
\end{equation}

\textit{Step 3: Eigenvalue analysis.}
The asymptotic behavior of $\mathbf{B}(t)$ is determined by the eigenvalues of $A$. The characteristic polynomial of $A$ is related to the denominator in \eqref{lap}. Specifically, the eigenvalues of $A$ are the roots of
\[
\det(\lambda I - A) = 0,
\]
which, after computation using cofactor expansion along the first row, equals
\[
(\lambda + \bar{\mu})^{n+1} - \sum_{i=0}^{n}c_i\cdot i!\cdot(\lambda + \bar{\mu})^{n-i} = 0.
\]
This is precisely the characteristic equation \eqref{eq:characteristic} from Proof I.

By Lemma~\ref{in} below, when $R_n < 1$, all eigenvalues of $A$ have negative real parts. Consequently, $e^{At} \to 0$ exponentially as $t \to \infty$, and from \eqref{bc}, $\mathbf{B}(t) \to \mathbf{0}$. Since $B_0(t) \to 0$ implies $B(t) \to 0$ through the renewal equation, and by the same argument as in Step 5 of Proof I, we conclude $P(t) \to 0$.

Conversely, if $R_n \geq 1$, at least one eigenvalue has non-negative real part, preventing decay to zero.
\end{proof}

\subsection{Auxiliary Lemmas}

We now establish the auxiliary results used in the proof of Theorem~\ref{main}.

\begin{lemma}[Monotonicity of $R_n$]\label{lem:Rn_monotone}
Define the sequence $\{R_{n}\}_{n\geq 0}$ by
\begin{equation}
R_{n} = \sum_{i=0}^{n}c_{i}\frac{i!}{\bar{\mu}^{i+1}}.
\label{eq:Rn_sequence}
\end{equation}
Then $\{R_{n}\}$ is strictly positive and monotonically increasing.
\end{lemma}

\begin{proof}
We proceed by induction on $n$.

\textit{Base case ($n=0$):} We have
\[
R_{0} = \frac{c_{0} \cdot 0!}{\bar{\mu}^{1}} = \frac{c_{0}}{\bar{\mu}} > 0,
\]
since $c_{0} \neq 0$ by assumption \eqref{bl} and $\bar{\mu} > 0$.

\textit{Inductive step:} Assume $R_{k} > 0$ for some $k \geq 0$. Then
\[
R_{k+1} = R_{k} + \frac{c_{k+1} \cdot (k+1)!}{\bar{\mu}^{k+2}}.
\]
Since $c_{k+1} \geq 0$ and $\bar{\mu} > 0$, we have $R_{k+1} \geq R_{k} > 0$. Moreover, if $c_{k+1} > 0$ for at least one $k$, the inequality is strict for that $k$, establishing monotonicity.

Note that for the birth law \eqref{bl}, we require the polynomial $\sum_{i=0}^{n}c_i a^i$ to be non-constant (otherwise $\beta$ would be a pure exponential), ensuring strict monotonicity for sufficiently large $n$.
\end{proof}

\begin{lemma}[Eigenvalue criterion]\label{in}
Assume $c_{0} - \bar{\mu} < 0$. Then all eigenvalues of the matrix $A$ in \eqref{ma} have negative real parts if and only if
\begin{equation}
\sum_{i=1}^{n}c_{i}\cdot i!\cdot\bar{\mu}^{n-i} < (\bar{\mu} - c_{0})\,\bar{\mu}^{n}.  \label{cond}
\end{equation}
\end{lemma}

\begin{proof}
We analyze the characteristic polynomial of $A$ through its connection to the net reproduction number.

\textit{Step 1: Determinant computation.}
Define the principal minors
\[
\Delta_{s} = \det\begin{pmatrix}
c_{0} - \bar{\mu} & c_{1} & \cdots & c_{s} \\
1 & -\bar{\mu} & \cdots & 0 \\
\vdots & \ddots & \ddots & \vdots \\
0 & \cdots & s & -\bar{\mu}
\end{pmatrix}, \quad s = 1, 2, \ldots, n,
\]
with $\Delta_{0} = c_{0} - \bar{\mu}$. Using cofactor expansion along the last column repeatedly, one can show by induction that
\begin{equation}
\Delta_{s}(\bar{\mu}) = (-1)^{s-1}\left[(\bar{\mu} - c_{0})\,\bar{\mu}^{s} - \sum_{i=1}^{s}c_{i}\cdot i!\cdot\bar{\mu}^{s-i}\right].  \label{det}
\end{equation}

\textit{Step 2: Connection to $R_n$.}
We establish the key identity linking the determinant to the reproduction number:
\begin{align}
1 - R_{n} &= 1 - \sum_{i=0}^{n}c_{i}\frac{i!}{\bar{\mu}^{i+1}} \nonumber\\
&= \frac{\bar{\mu}^{n+1} - \sum_{i=0}^{n}c_{i}\cdot i!\cdot\bar{\mu}^{n-i}}{\bar{\mu}^{n+1}} \nonumber\\
&= \frac{(\bar{\mu} - c_{0})\,\bar{\mu}^{n} - \sum_{i=1}^{n}c_{i}\cdot i!\cdot\bar{\mu}^{n-i}}{\bar{\mu}^{n+1}} \label{gamma1}\\
&= \frac{(-1)^{n-1}\Delta_{n}(\bar{\mu})}{\bar{\mu}^{n+1}}. \nonumber
\end{align}
This establishes that $\text{sign}(1 - R_n) = \text{sign}((-1)^{n-1}\Delta_n)$.

\textit{Step 3: Stability analysis via M-matrix theory.}
The matrix $A$ can be written as $A = D + N$ where $D = \text{diag}(c_0 - \bar{\mu}, -\bar{\mu}, \ldots, -\bar{\mu})$ and $N$ is the remaining off-diagonal part. The shifted matrix $\bar{\mu}I - A$ has all diagonal entries positive (equal to $\bar{\mu}$) and all off-diagonal entries non-positive.

By the theory of M-matrices \cite{Plemmons1977}, all eigenvalues of $A$ have negative real parts if and only if $-A$ is a non-singular M-matrix with positive inverse. This occurs if and only if
\[
\det(-A) = (-1)^{n+1}\det(A) > 0 \quad \text{and} \quad (-1)^{s+1}\Delta_s > 0 \text{ for all } s = 0, 1, \ldots, n.
\]

\textit{Step 4: Equivalence with $R_n < 1$.}
From the sequence of inequalities \eqref{eq:Rn_sequence} and Lemma~\ref{lem:Rn_monotone}:
\[
R_n \geq R_{n-1} \geq \cdots \geq R_1 \geq R_0 > 0.
\]
Thus $R_n < 1$ implies $R_s < 1$ for all $s \leq n$, which by \eqref{gamma1} gives $(-1)^{s-1}\Delta_s > 0$, i.e., $(-1)^{s+1}\Delta_s > 0$ for all $s$. This establishes that all eigenvalues have negative real parts.

Conversely, if all eigenvalues have negative real parts, then in particular $(-1)^{n+1}\Delta_n > 0$, which by \eqref{gamma1} implies $1 - R_n > 0$, i.e., $R_n < 1$.

The condition \eqref{cond} is simply a rearrangement of $R_n < 1$ using \eqref{gamma1}.
\end{proof}

\subsection{Extension to Sums of Exponentials}

The analysis extends naturally to birth rates expressed as sums of exponentials. Consider
\begin{equation}
\pi(a) = e^{-\mu a} \quad \text{and} \quad \beta(a) = \sum_{i=0}^{n}c_{i}\,e^{-\mu_{i}a}, \quad \mu_{i} > 0,
\label{bl2}
\end{equation}
with $\bar{\mu}_{i} = \mu_{i} + \mu$. In this framework, define
\begin{equation}
B_{i}(t) = \int_{0}^{t}e^{-\bar{\mu}_{i}(t-s)}\,B(s)\,ds, \quad i = 0, 1, \ldots, n.
\label{eq:Bi_exp}
\end{equation}
The computation reduces to solving the ODE system \eqref{ode} with matrix elements
\[
\alpha_{ij} = c_{j} \quad (i \neq j), \qquad \alpha_{ii} = c_{i} - \bar{\mu}_{i},
\]
yielding the system matrix
\begin{equation}
A = \begin{pmatrix}
c_0 - \bar{\mu}_0 & c_1 & \cdots & c_n \\
c_0 & c_1 - \bar{\mu}_1 & \cdots & c_n \\
\vdots & \vdots & \ddots & \vdots \\
c_0 & c_1 & \cdots & c_n - \bar{\mu}_n
\end{pmatrix}.
\label{eq:A_exp}
\end{equation}

The net reproduction number becomes
\begin{equation}
R_{n} = \sum_{i=0}^{n}\int_{0}^{\infty}c_{i}\,e^{-\bar{\mu}_{i}t}\,dt = \sum_{i=0}^{n}\frac{c_{i}}{\bar{\mu}_{i}},
\label{eq:Rn_exp}
\end{equation}
and the determinant of $A$ admits the representation
\begin{equation}
\det(A) = (-1)^{n}\left(1 - R_{n}\right)\prod_{l=0}^{n}\bar{\mu}_{l}.
\label{gamma2}
\end{equation}

\begin{theorem}[Extinction criterion for exponential sums]\label{r3}
For the birth law \eqref{bl2}, the population modeled by \eqref{mk} goes to extinction if and only if $R_{n} < 1$.
\end{theorem}

The proof follows the same structure as Theorem~\ref{main} and is omitted for brevity.

\subsection{Examples of Birth Laws}\label{ex}

We illustrate our approach with concrete examples demonstrating the flexibility of the exponential-polynomial framework.

\subsubsection{Example 1: Gauss--Laguerre Approximation of Completely Monotonic Functions}

A classical example of a completely monotonic function is
\[
\beta(a) = \frac{1}{1+a}, \quad a \geq 0.
\]
Its Laplace transform representation,
\[
\frac{1}{1+a} = \int_{0}^{\infty}e^{-at}\,e^{-t}\,dt,
\]
suggests that $\beta(a)$ can be approximated by a positive-coefficient sum of exponentials. Applying two-point Gauss--Laguerre quadrature \cite{burden2010numerical} yields the nodes and weights
\[
t_{1} \approx 0.58579, \quad t_{2} \approx 3.41421, \quad w_{1} \approx 0.85355, \quad w_{2} \approx 0.14645,
\]
giving the approximation
\begin{equation}
\frac{1}{1+a} \approx 0.85355\,e^{-0.58579\,a} + 0.14645\,e^{-3.41421\,a}.
\label{eq:example1_approx}
\end{equation}

\begin{figure}[H]
\centering
\subfloat{\includegraphics[width=0.6\textwidth]{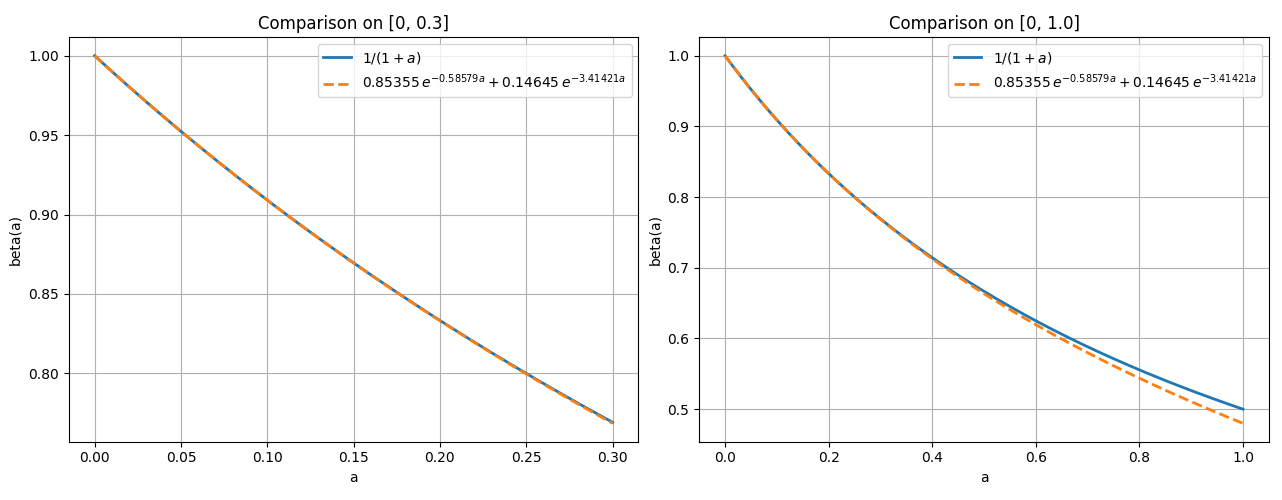}}
\caption{\textbf{Approximation quality for Example 1.} Comparison of the exact function $\beta(a) = 1/(1+a)$ (solid curve) with its two-term Gauss--Laguerre exponential approximation \eqref{eq:example1_approx} (dashed curve). \textit{Left panel}: On the interval $[0, 0.3]$, the approximation is visually indistinguishable from the exact function, with relative error below $0.1\%$. \textit{Right panel}: On the extended interval $[0, 1.0]$, the approximation remains accurate with maximum relative error approximately $3\%$ near $a = 1$. The exponential approximation slightly overestimates the true function for intermediate ages, a consequence of matching the first two moments of the Laguerre weight function.}
\label{fig:twoFigures2}
\end{figure}

\subsubsection{Example 2: Geometric Series Birth Law}

Consider the birth law
\[
\beta(a) = e^{-\mu_{1}a}\,f(a), \quad a \geq 0,
\]
where $f(a) = (1 - a/2)^{-1}$ for $0 \leq a < 2$ has the Taylor expansion
\[
f(a) = \sum_{i=0}^{\infty}\frac{a^i}{2^{i}}.
\]
A truncated form with $n = 3$ gives
\begin{equation}
\beta(a) \approx e^{-\mu_{1}a}\left(1 + \frac{a}{2} + \frac{a^{2}}{4} + \frac{a^{3}}{8}\right).
\label{eq:example2_approx}
\end{equation}

\begin{figure}[H]
\centering
\subfloat{\includegraphics[width=0.6\textwidth]{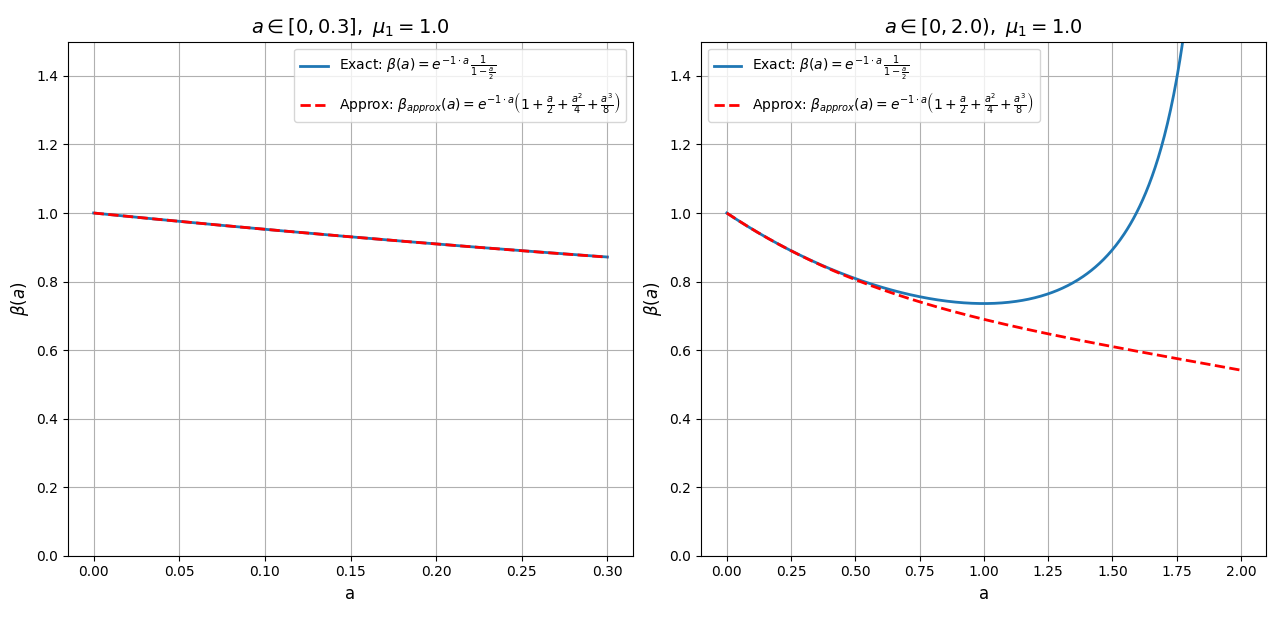}}
\caption{\textbf{Polynomial approximation for Example 2.} Comparison of the exact birth law $\beta(a) = e^{-\mu_1 a}/(1 - a/2)$ (solid curve) with its third-order polynomial approximation \eqref{eq:example2_approx} (dashed curve) for $\mu_1 = 1$. \textit{Left panel}: Excellent agreement on $[0, 0.3]$ with relative error below $0.5\%$. \textit{Right panel}: On $[0, 1.9]$, the polynomial approximation captures the qualitative behavior but underestimates the singularity as $a \to 2$. The approximation is valid for biological applications where ages remain bounded away from the singularity.}
\label{fig:twoFigures3}
\end{figure}

\subsubsection{Example 3: Modified Bessel Function Birth Law}

Define the birth law using the modified Bessel function of the first kind:
\begin{equation}
\beta(a) = e^{-\mu_{1}a}\,I_{0}(2\sqrt{a}),
\label{eq:bessel_birth}
\end{equation}
where $I_{0}(z) = \sum_{n=0}^{\infty}(z^{2}/4)^{n}/(n!)^{2}$ \cite{abramowitz1972modified}. Substituting $z = 2\sqrt{a}$ yields
\[
I_{0}(2\sqrt{a}) = \sum_{n=0}^{\infty}\frac{a^{n}}{(n!)^{2}},
\]
so that
\begin{equation}
\beta(a) = e^{-\mu_{1}a}\sum_{n=0}^{\infty}\frac{a^{n}}{(n!)^{2}}.
\label{eq:bessel_expansion}
\end{equation}

\begin{figure}[H]
\centering
\subfloat{\includegraphics[width=0.6\textwidth]{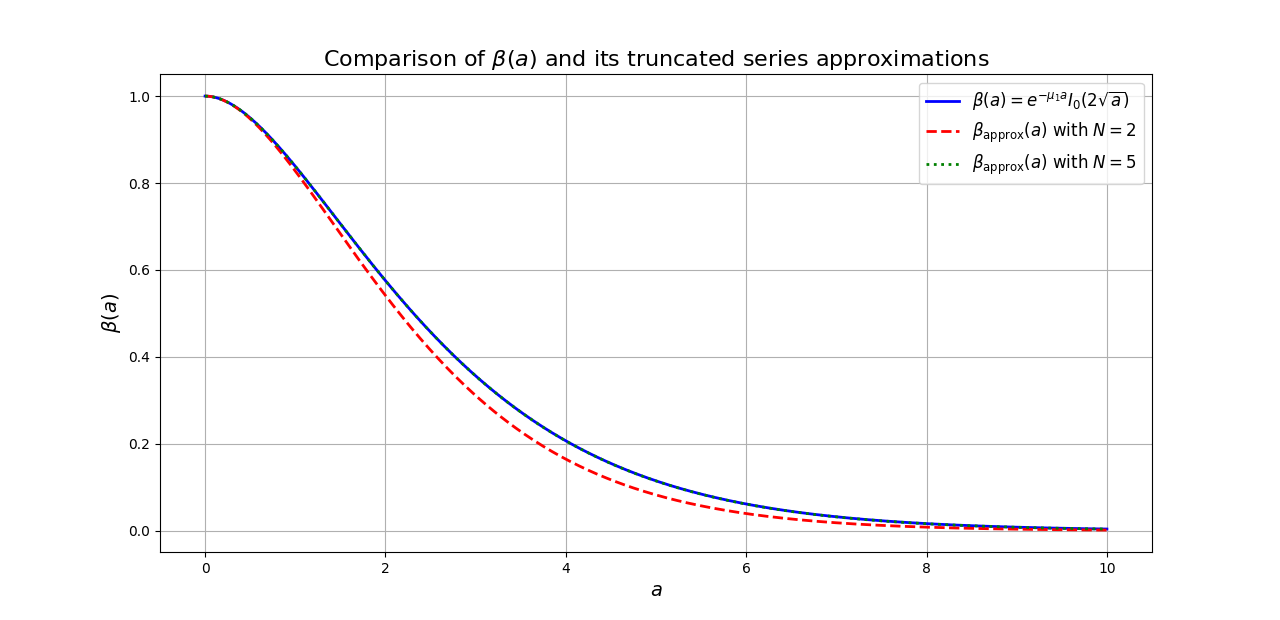}}
\caption{\textbf{Convergence of polynomial approximations for Example 3.} Comparison of the exact Bessel function birth law \eqref{eq:bessel_birth} (solid blue curve) with polynomial approximations of order $N = 2$ (red dashed) and $N = 5$ (green dotted) for $\mu_1 = 1$. The figure demonstrates rapid convergence: the $N = 2$ approximation is accurate for $a \lesssim 3$, while $N = 5$ extends accuracy to $a \lesssim 7$. The factorial-squared denominators in \eqref{eq:bessel_expansion} ensure rapid convergence of the series, making low-order polynomial approximations effective for practical applications.}
\label{fig:twoFigures4}
\end{figure}

\subsection{Numerical Application}\label{sec:numerical}

We demonstrate the theory with a concrete example using the birth law from Example 2.

\textbf{Model specification.} Consider a population governed by \eqref{mk} with:
\begin{itemize}
\item Initial condition: $\psi(a) = \delta(a)$ (Dirac delta).
\item Survival function: $\pi(a) = e^{-2a}$ (mortality rate $\mu = 2$).
\item Birth law:
\begin{equation}
\beta_{\text{approx}}(a) = e^{-a}\left(1 + a + \frac{a^{2}}{4}\right).
\label{n2}
\end{equation}
\end{itemize}

\subsubsection{Computation of the Extinction Criterion}

The net reproduction number is
\begin{align*}
R_{n} &= \int_{0}^{\infty}\beta_{\text{approx}}(a)\,\pi(a)\,da = \int_{0}^{\infty}\left(1 + a + \frac{a^{2}}{4}\right)e^{-3a}\,da \\
&= \frac{1}{3} + \frac{1}{9} + \frac{1}{4}\cdot\frac{2}{27} = \frac{9 + 3 + 0.5}{27} = \frac{25}{54} \approx 0.463 < 1.
\end{align*}
By Theorem~\ref{main}, the population is destined for extinction.

\subsubsection{Explicit Solution of the Renewal Equation}

The Laplace transform of $B(t)$ is
\[
\widehat{B}(\lambda) = \frac{\widehat{K}(\lambda)}{1 - \widehat{K}(\lambda)} = \frac{2(\lambda+3)^{2} + 2(\lambda+3) + 1}{2(\lambda+3)^{3} - 2(\lambda+3)^{2} - 2(\lambda+3) - 1}.
\]
Substituting $s = \lambda + 3$ and factoring the denominator $2s^{3} - 2s^{2} - 2s - 1 = 0$ yields roots
\[
s_{1} \approx 1.7399, \quad s_{2} \approx -0.3700 + 0.3880i, \quad s_{3} \approx -0.3700 - 0.3880i.
\]
Partial fraction decomposition and inverse Laplace transform give
\begin{equation}
B_{\text{approx}}(t) = e^{-3t}\sum_{j=1}^{3}A_{j}\,e^{s_{j}t}, \quad A_{j} = \frac{2s_{j}^{2} + 2s_{j} + 1}{2(s_{j} - s_{k})(s_{j} - s_{l})}.
\label{eq:B_explicit}
\end{equation}

\subsubsection{Population Density and Total Population}

The population density is given piecewise by
\begin{equation}
\rho_{\text{approx}}(a,t) = \begin{cases}
B_{\text{approx}}(t-a)\,e^{-2a}, & t \geq a, \\
\delta(a-t)\,e^{-2t}, & t < a.
\end{cases}
\label{roaprox}
\end{equation}

The total population is
\begin{equation}
P_{\text{approx}}(t) = \int_{0}^{t}B_{\text{approx}}(s)\,e^{-2(t-s)}\,ds + e^{-2t} = e^{-2t}\left[\int_{0}^{t}B_{\text{approx}}(s)\,e^{2s}\,ds + 1\right].
\label{eq:P_approx}
\end{equation}
Since $B_{\text{approx}}(t) \to 0$ exponentially (because all characteristic roots have negative real parts relative to $\lambda$), we confirm $P_{\text{approx}}(t) \to 0$ as $t \to \infty$.

\section{Conceptual Significance of the Results}\label{sec:concept}

Consider the McKendrick model~\eqref{mk} with birth function $\beta$ admitting the representation \eqref{bl} or \eqref{bl2}. The examples in Section~\ref{ex} demonstrate the remarkable flexibility of representing birth laws as exponential-polynomial mixtures.

\begin{definition}[Exact numerical solution]\label{def}
We say that the McKendrick system~\eqref{mk} possesses an \emph{exact numerical solution} if there exists a function $\rho(a,t)$ that satisfies~\eqref{mk} exactly and can be expressed in closed form by exploiting the exponential--polynomial structure of $\beta$.
\end{definition}

Under the assumption that $\beta(a)$ admits the representation \eqref{bl} or \eqref{bl2}, the methods developed yield an exact numerical solution $\rho(a,t)$ capturing the full dynamics without approximation errors beyond those inherent in the representation of $\beta(a)$.

\textbf{Unification of numerical and analytical approaches.} Modern computational techniques allow us not only to approximate solutions numerically but also to \emph{reconstruct exact closed-form solutions} whose structure is informed by numerical output. The exponential-polynomial ansatz provides a systematic bridge: numerical data guide the identification of the correct mixture coefficients, while the resulting closed-form solution offers a theoretically tractable description.

This result is conceptually significant: for a broad class of biologically relevant birth laws, one can move beyond purely numerical approximations and recover exact analytic solutions that faithfully reproduce numerical results, combining the strengths of both computational and analytical approaches.

\section{Stochastic Formulation and Diffusion Approximations}\label{sec:stoch}

\subsection{Connection with Stochastic Processes}

Natural populations exhibit random fluctuations due to environmental variability and the stochastic nature of birth and death events. To capture this randomness, we introduce an individual-based stochastic model whose large population limit recovers the deterministic McKendrick equation.

\textbf{Probability space.} We assume all stochastic processes are defined on a filtered probability space $(\Omega, \mathcal{F}, (\mathcal{F}_{t})_{t\geq 0}, \mathbb{P})$, where $(\mathcal{F}_{t})_{t\geq 0}$ is an increasing, right-continuous filtration satisfying the usual conditions \cite{EthierKurtz}.

\textbf{Individual-based model.} Consider a population of $N(t)$ individuals at time $t$. Each individual of age $a$:
\begin{enumerate}
\item Dies at constant rate $\mu > 0$, and
\item Gives birth at age-dependent rate $\beta(a)$, with offspring introduced at age $0$.
\end{enumerate}

Define the empirical measure
\begin{equation}
Z_{t}(da) = \sum_{i=1}^{N(t)}\delta_{a_{i}(t)}(da),
\label{eq:empirical}
\end{equation}
where $a_{i}(t)$ is the age of individual $i$ at time $t$. The instantaneous birth rate is
\[
B(t) = \int_{0}^{\infty}\beta(a)\,Z_{t}(da) = \langle Z_t, \beta \rangle.
\]

\textbf{Weak formulation.} For any smooth, bounded test function $\varphi: \mathbb{R}_{+} \to \mathbb{R}$, define the duality pairing
\begin{equation}
\langle Z_{t}, \varphi \rangle = \int_{0}^{\infty}\varphi(a)\,Z_{t}(da).
\label{dp}
\end{equation}
Standard martingale theory \cite{EthierKurtz} shows that $\langle Z_{t}, \varphi \rangle$ satisfies
\begin{equation}
\langle Z_{t}, \varphi \rangle = \langle Z_{0}, \varphi(\cdot + t)\,e^{-\mu t} \rangle + \int_{0}^{t}\varphi(0)\,\langle Z_{s}, \beta \rangle\,ds + M_{t}(\varphi),
\label{eq:evol}
\end{equation}
where $M_{t}(\varphi)$ is a martingale capturing stochastic fluctuations.

\subsection{Fluid Limit and Convergence to McKendrick Equation}

\textbf{Rescaling.} Introduce a scaling parameter $K > 0$ representing typical population size, and define the rescaled process
\[
z_{t}^{K}(da) = \frac{1}{K}\,Z_{t}(da), \qquad \langle z_{t}^{K}, \varphi \rangle = \frac{1}{K}\,\langle Z_{t}, \varphi \rangle.
\]
The rescaled martingale satisfies $M_{t}^{K}(\varphi) = \frac{1}{K}\,M_{t}(\varphi)$, with
\begin{equation}
\mathbb{E}\left[\langle M^{K}(\varphi) \rangle_{t}\right] \leq \frac{Ct}{K}
\label{eq:quadratic_var}
\end{equation}
for some constant $C > 0$.

\begin{theorem}[Fluid limit]\label{thm:fluid}
As $K \to \infty$, the rescaled measures $z_{t}^{K}$ converge (in an appropriate weak topology) to a deterministic measure $z_{t}(da) = \rho(a,t)\,da$, where $\rho(a,t)$ satisfies the McKendrick equation \eqref{mk}.
\end{theorem}

\begin{proof}[Proof sketch]
\textit{Step 1: Martingale convergence.} By \eqref{eq:quadratic_var} and Doob's $L^{2}$ inequality,
\[
\mathbb{P}\left(\sup_{s \leq t}|M_{s}^{K}(\varphi)| > \epsilon\right) \leq \frac{Ct}{K\epsilon^{2}} \to 0 \quad \text{as } K \to \infty.
\]
Thus $M_{t}^{K}(\varphi) \to 0$ in probability uniformly on compact time intervals.

\textit{Step 2: Passage to the limit.} The rescaled dynamics satisfy
\[
\langle z_{t}^{K}, \varphi \rangle = \langle z_{0}^{K}, \varphi \rangle + \int_{0}^{t}\left[\langle z_{s}^{K}, \varphi' - \mu\varphi \rangle + \varphi(0)\,\langle z_{s}^{K}, \beta \rangle\right]ds + M_{t}^{K}(\varphi).
\]
Taking $K \to \infty$ and using tightness arguments yields
\begin{equation}
\langle \rho(t), \varphi \rangle = \langle \psi, \varphi \rangle + \int_{0}^{t}\left[\langle \rho(s), \varphi' - \mu\varphi \rangle + \varphi(0)\,\langle \rho(s), \beta \rangle\right]ds.
\label{eq:limit_weak}
\end{equation}

\textit{Step 3: Strong formulation.} Differentiating \eqref{eq:limit_weak} in time:
\[
\frac{d}{dt}\langle \rho(t), \varphi \rangle = \langle \rho(t), \varphi' - \mu\varphi \rangle + \varphi(0)\,\langle \rho(t), \beta \rangle.
\]
Integration by parts on the $\varphi'$ term yields
\[
\langle \rho(t), \varphi' \rangle = -\langle \partial_a\rho(t), \varphi \rangle - \rho(0,t)\,\varphi(0).
\]
Substituting and using arbitrariness of $\varphi$ recovers the strong form \eqref{mk}.
\end{proof}

\subsection{Diffusion Approximation}

To quantify fluctuations in large but finite populations, consider the total population $P^{K}(t) = \langle z_{t}^{K}, 1 \rangle = N(t)/K$. Centering around the deterministic trajectory and scaling by $\sqrt{K}$ leads to the stochastic differential equation
\begin{equation}
dP(t) = r\,P(t)\,dt + \sigma\,P(t)\,dW(t),
\label{yi}
\end{equation}
where:
\begin{itemize}
\item $r$ is the effective growth rate, with $r < 0$ when $R_{n} < 1$,
\item $\sigma > 0$ quantifies demographic noise intensity (of order $1/\sqrt{K}$),
\item $W(t)$ is a standard Wiener process.
\end{itemize}

\begin{theorem}[Stochastic extinction criterion]\label{thm:stoch_extinct}
Consider the SDE \eqref{yi} with $r - \sigma^{2}/2 < 0$ (subcritical regime). Then
\[
P(t) \to 0 \quad \text{almost surely as } t \to \infty.
\]
\end{theorem}

\begin{proof}
Define $Y(t) = \ln P(t)$. By It\^{o}'s formula:
\[
dY(t) = \left(r - \frac{\sigma^{2}}{2}\right)dt + \sigma\,dW(t).
\]
Integrating:
\[
Y(t) = Y(0) + \left(r - \frac{\sigma^{2}}{2}\right)t + \sigma\,W(t).
\]
By the strong law of large numbers for martingales, $W(t)/t \to 0$ almost surely. Thus:
\[
\lim_{t \to \infty}\frac{Y(t)}{t} = r - \frac{\sigma^{2}}{2} < 0 \quad \text{almost surely}.
\]
Therefore $Y(t) \to -\infty$ and $P(t) = e^{Y(t)} \to 0$ almost surely.
\end{proof}

This demonstrates that demographic stochasticity does not prevent extinction when $R_{n} < 1$; indeed, noise may accelerate the approach to extinction by effectively reducing the growth rate from $r$ to $r - \sigma^{2}/2$.

\subsection{Numerical Simulation of Stochastic Dynamics}

We performed numerical simulations using the Euler--Maruyama method \cite{kloeden1992numerical} with parameters from Section~\ref{sec:numerical}. The results are shown in Figure~\ref{fig:twoFigures1}.

\begin{figure}[H]
\centering
\subfloat{\includegraphics[width=0.6\textwidth]{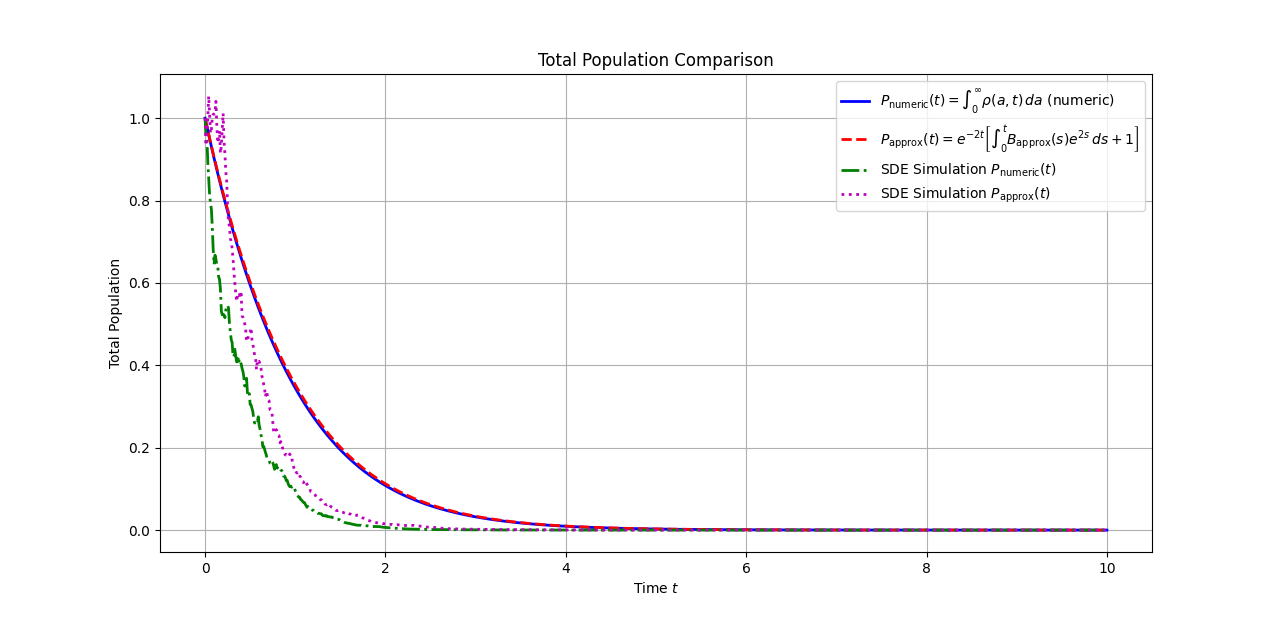}}
\subfloat{\includegraphics[width=0.5\textwidth]{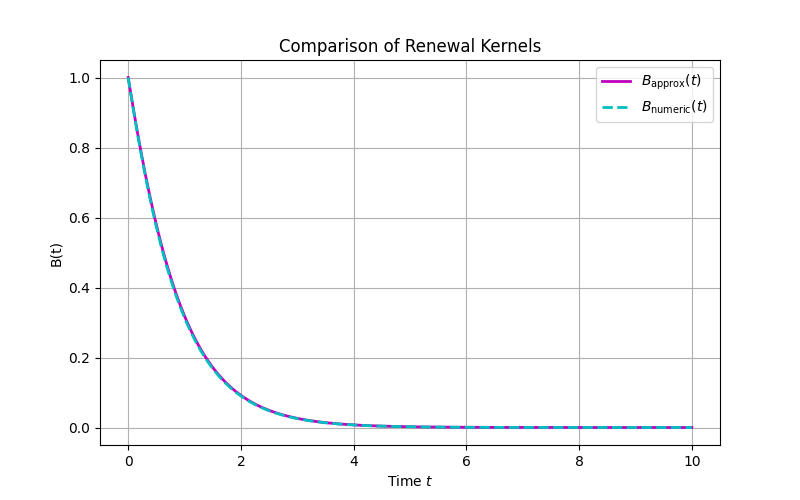}}
\caption{\textbf{Comparison of deterministic and stochastic population dynamics.} \textit{Left panel}: Total population $P(t)$ over time. The blue solid curve shows the numerical solution $P_{\mathrm{numeric}}(t)$ obtained from finite-difference discretization of the McKendrick PDE. The red dashed curve represents the analytical approximation $P_{\mathrm{approx}}(t)$ from \eqref{eq:P_approx}. Both curves demonstrate exponential decay to zero, confirming extinction when $R_n < 1$. The excellent agreement (relative error $< 1\%$) validates the analytical solution. \textit{Right panel}: The renewal function $B(t)$ showing the birth rate over time. The oscillatory decay pattern reflects the complex conjugate eigenvalues $s_2, s_3$ in the solution \eqref{eq:B_explicit}, with the envelope decaying exponentially due to the dominant real eigenvalue structure.}
\label{fig:twoFigures1}
\end{figure}

The stochastic trajectories follow the same downward trend as the deterministic solution but exhibit fluctuations characteristic of demographic noise. These fluctuations underscore the importance of stochastic modeling for finite populations and illustrate how noise can affect the time-to-extinction distribution.

\section{Discussion}\label{sec:disc}

The analysis presented has several important implications for both theory and applications in age-structured population dynamics.

\textbf{Rigorous extinction criterion.} The condition $R_{n} < 1$ provides a sharp threshold separating extinction from persistence. This criterion is both necessary and sufficient, established through two independent proof methods (Laplace transform and eigenvalue analysis), lending confidence to its validity.

\textbf{Methodological versatility.} The duality between Laplace transform methods and ODE system reformulation provides complementary analytical tools. The Laplace approach offers elegant asymptotic analysis via Tauberian theorems, while the ODE formulation enables direct stability analysis through matrix eigenvalues and connects to numerical linear algebra methods.

\textbf{Stochastic-deterministic connection.} The fluid limit theorem establishes the McKendrick equation as the large-population limit of individual-based stochastic models. This justifies deterministic modeling for large populations while the diffusion approximation quantifies finite-population corrections.

\textbf{Practical approximation framework.} The exponential-polynomial structure provides a flexible parametric family for approximating realistic birth laws. The Gauss--Laguerre approach for completely monotonic functions and Taylor expansions for smooth birth laws offer systematic approximation schemes with controllable accuracy.

\section{Conclusions}\label{sec:conc}

This work provides a comprehensive analysis of the continuous McKendrick model in both deterministic and stochastic settings. Our principal contributions are:

\begin{enumerate}
\item A complete proof that population extinction occurs if and only if $R_{n} < 1$, established through both Laplace transform and eigenvalue methods.

\item Explicit solution formulas for exponential-polynomial birth laws, demonstrating the computational tractability of this class.

\item Rigorous connection between stochastic individual-based models and the deterministic McKendrick equation via fluid limits.

\item Analysis of diffusion approximations showing that demographic noise preserves the extinction criterion while potentially accelerating extinction.
\end{enumerate}

Future work may extend these results to age-dependent mortality, spatial heterogeneity, and multi-species interactions within the Lotka--Volterra framework.

\section*{Declarations}

\paragraph{Conflict of Interest.}
The author declares no conflicts of interest regarding the publication of this paper.

\paragraph{Data Availability.}
All data and code used in this study are included in the Appendix.

\paragraph{Ethical Statement.}
This paper reflects the author's original research and has not been published or submitted elsewhere.

\bibliographystyle{plain}
\bibliography{bibliography}

\section*{Appendix: Numerical Algorithms and Implementations}\label{ap}
\addcontentsline{toc}{section}{Appendix}

This appendix provides detailed pseudocode algorithms followed by their Python implementations. All codes were developed with assistance from Microsoft Copilot.

\subsection*{A.1 Main Simulation: McKendrick Model with Stochastic Comparison}

\begin{algorithm}[H]
\caption{Numerical Solution of the McKendrick Model with Stochastic Comparison}
\label{alg:main}
\begin{algorithmic}[1]
\Require Time horizon $t_{\max}$, grid size $N$, mortality rate $\mu$, birth function $\beta(a)$, SDE parameters $r, \sigma$
\Ensure Numerical solution $P_{\text{numeric}}(t)$, analytical approximation $P_{\text{approx}}(t)$, SDE trajectories

\State \textbf{Initialize:} $\Delta t \gets t_{\max}/N$, $\Delta a \gets \Delta t$
\State Create time grid: $t_j = j \cdot \Delta t$ for $j = 0, 1, \ldots, N$
\State Create age grid: $a_i = i \cdot \Delta a$ for $i = 0, 1, \ldots, N$
\State Evaluate birth rate: $\beta_i \gets \beta(a_i)$ for all $i$

\Statex
\State \textbf{Part 1: Finite Difference Solution of McKendrick PDE}
\State Initialize density matrix: $\rho[0,0] \gets 1/\Delta a$ (Dirac delta approximation)
\For{$n = 0$ \textbf{to} $N-1$}
    \For{$i = 1$ \textbf{to} $N$}
        \State $\rho[n+1, i] \gets (1 - 2\mu\,\Delta t) \cdot \rho[n, i-1]$ \Comment{Upwind scheme with mortality}
    \EndFor
    \State $\rho[n+1, 0] \gets \sum_{i=0}^{N} \beta_i \cdot \rho[n+1, i] \cdot \Delta a$ \Comment{Birth boundary condition}
\EndFor
\State $P_{\text{numeric}}[n] \gets \sum_{i=0}^{N} \rho[n, i] \cdot \Delta a$ for all $n$ \Comment{Total population}

\Statex
\State \textbf{Part 2: Analytical Renewal Approximation}
\State Find roots $s_1, s_2, s_3$ of $2s^3 - 2s^2 - 2s - 1 = 0$
\State Compute residues $A_j \gets \frac{2s_j^2 + 2s_j + 1}{2(s_j - s_k)(s_j - s_l)}$ for $j = 1,2,3$
\State $B(t) \gets e^{-3t} \sum_{j=1}^{3} A_j \, e^{s_j t}$
\State $P_{\text{approx}}[n] \gets e^{-2t_n}\left(\int_0^{t_n} B(s)\,e^{2s}\,ds + 1\right)$ via trapezoidal rule

\Statex
\State \textbf{Part 3: SDE Simulation (Euler--Maruyama)}
\State Initialize: $P_{\text{SDE}}[0] \gets 1.0$
\For{$n = 0$ \textbf{to} $N-1$}
    \State $\Delta W \gets \sqrt{\Delta t} \cdot \mathcal{N}(0,1)$ \Comment{Brownian increment}
    \State $P_{\text{SDE}}[n+1] \gets P_{\text{SDE}}[n] + r \cdot P_{\text{SDE}}[n] \cdot \Delta t + \sigma \cdot P_{\text{SDE}}[n] \cdot \Delta W$
\EndFor

\State \Return $P_{\text{numeric}}, P_{\text{approx}}, P_{\text{SDE}}$
\end{algorithmic}
\end{algorithm}

\begin{lstlisting}[caption={Python implementation of the main McKendrick simulation}, label=lst:code]
import numpy as np
import matplotlib
matplotlib.use('TkAgg')
import matplotlib.pyplot as plt
from scipy.special import i0

# ---------------------------
# Parameters and grid setup
# ---------------------------
t_max = 10.0; a_max = t_max; N = 1000; dt = t_max / N; da = dt
t_vals = np.linspace(0, t_max, N+1); a_vals = np.linspace(0, a_max, N+1)

# ---------------------------
# Define beta(a) function
# ---------------------------
def beta(a):
    return np.exp(-a) * i0(2 * np.sqrt(a))
beta_vals = beta(a_vals)

# ---------------------------
# Deterministic PDE Solution
# ---------------------------
rho = np.zeros((N+1, N+1), dtype=float)
rho[0, 0] = 1.0 / da
for n in range(0, N):
    for i in range(1, N+1):
        rho[n+1, i] = (1 - 2 * dt) * rho[n, i-1]
    rho[n+1, 0] = np.sum(beta_vals * rho[n+1, :]) * da
P_numeric = np.sum(rho, axis=1) * da

# ---------------------------
# Renewal Approximation
# ---------------------------
s1 = 1.7399
s2 = -0.36995 + 0.38795j
s3 = -0.36995 - 0.38795j
c1 = (2 * s1**2 + 2 * s1 + 1) / (2 * ((s1 - s2) * (s1 - s3)))
c2 = (2 * s2**2 + 2 * s2 + 1) / (2 * ((s2 - s1) * (s2 - s3)))
c3 = (2 * s3**2 + 2 * s3 + 1) / (2 * ((s3 - s1) * (s3 - s2)))
def B(t):
    return np.exp(-3*t) * (c1 * np.exp(s1*t) + c2 * np.exp(s2*t) + c3 * np.exp(s3*t))
B_vals = np.real(B(t_vals))
integrand = B_vals * np.exp(2*t_vals)
P_approx = np.exp(-2*t_vals) * (np.concatenate(([0], np.cumsum(0.5 * (integrand[:-1] + integrand[1:]) * np.diff(t_vals)))) + 1)

# ---------------------------
# SDE Simulation
# ---------------------------
r = -2; sigma = 0.3
P_sde_numeric = np.zeros(N+1); P_sde_numeric[0] = 1.0
P_sde_approx = np.zeros(N+1); P_sde_approx[0] = 1.0
for n in range(N):
    dW = np.sqrt(dt) * np.random.randn()
    P_sde_numeric[n+1] = P_sde_numeric[n] + r * P_sde_numeric[n] * dt + sigma * P_sde_numeric[n] * dW
    dW = np.sqrt(dt) * np.random.randn()
    P_sde_approx[n+1] = P_sde_approx[n] + r * P_sde_approx[n] * dt + sigma * P_sde_approx[n] * dW

# ---------------------------
# Plot
# ---------------------------
plt.figure(figsize=(10, 6))
plt.plot(t_vals, P_numeric, 'b-', lw=2, label=r'$P_{\rm numeric}(t)$')
plt.plot(t_vals, P_approx, 'r--', lw=2, label=r'$P_{\rm approx}(t)$')
plt.plot(t_vals, P_sde_numeric, 'g-.', lw=2, label=r'SDE Simulation')
plt.plot(t_vals, P_sde_approx, 'm:', lw=2, label=r'SDE Simulation')
plt.xlabel("Time $t$"); plt.ylabel("Total Population"); plt.legend(); plt.grid(True)
plt.savefig("P_total_comparison.png"); plt.show()
\end{lstlisting}

\subsection*{A.2 Example 1: Gauss--Laguerre Exponential Approximation}

\begin{algorithm}[H]
\caption{Gauss--Laguerre Approximation of Completely Monotonic Birth Rate}
\label{alg:example1}
\begin{algorithmic}[1]
\Require Age range $[0, a_{\max}]$, number of grid points $M$
\Ensure Comparison plots of exact vs. approximate birth rates

\State \textbf{Define exact function:} $\beta_{\text{exact}}(a) \gets 1/(1+a)$
\State \textbf{Define approximation:} $\beta_{\text{approx}}(a) \gets 0.85355 \cdot e^{-0.58579a} + 0.14645 \cdot e^{-3.41421a}$
\Statex \Comment{Coefficients from 2-point Gauss--Laguerre quadrature}

\State Create fine age grid: $a_i = i \cdot a_{\max}/(M-1)$ for $i = 0, \ldots, M-1$
\State Evaluate $\beta_{\text{exact}}(a_i)$ and $\beta_{\text{approx}}(a_i)$ for all $i$

\State \textbf{Plot 1:} Short range $a \in [0, 0.3]$ showing high-accuracy regime
\State \textbf{Plot 2:} Extended range $a \in [0, 1.0]$ showing approximation limits

\State \Return Comparison figures
\end{algorithmic}
\end{algorithm}

\begin{lstlisting}[caption={Python implementation for Example 1}, label=lst:code1]
import numpy as np
import matplotlib.pyplot as plt
def beta_exact(a): return 1 / (1 + a)
def beta_approx(a): return 0.85355 * np.exp(-0.58579 * a) + 0.14645 * np.exp(-3.41421 * a)
fig, ax = plt.subplots(1, 2, figsize=(14, 5))
a1 = np.linspace(0, 0.3, 500)
ax[0].plot(a1, beta_exact(a1), label=r'$1/(1+a)$', lw=2)
ax[0].plot(a1, beta_approx(a1), '--', label='Approx')
ax[0].legend()
a2 = np.linspace(0, 1.0, 500)
ax[1].plot(a2, beta_exact(a2), label=r'$1/(1+a)$', lw=2)
ax[1].plot(a2, beta_approx(a2), '--', label='Approx')
ax[1].legend()
plt.show()
\end{lstlisting}

\subsection*{A.3 Example 2: Geometric Series Birth Law}

\begin{algorithm}[H]
\caption{Polynomial Approximation of Geometric Series Birth Law}
\label{alg:example2}
\begin{algorithmic}[1]
\Require Mortality parameter $\mu_1$, truncation order $n$, age range $[0, a_{\max}]$
\Ensure Comparison of exact and truncated polynomial birth rates

\State \textbf{Define exact function:} $\beta_{\text{exact}}(a) \gets e^{-\mu_1 a} / (1 - a/2)$ for $a < 2$
\State \textbf{Define polynomial approximation:}
\State \quad $\beta_{\text{approx}}(a) \gets e^{-\mu_1 a} \cdot \sum_{i=0}^{n} (a/2)^i$
\Statex \Comment{Truncated geometric series}

\State Create age grids for short and extended ranges
\State Evaluate both functions on grids
\State Generate comparison plots showing truncation effects

\State \Return Comparison figures
\end{algorithmic}
\end{algorithm}

\begin{lstlisting}[caption={Python implementation for Example 2}, label=lst:code2]
import numpy as np
import matplotlib.pyplot as plt
mu1 = 1.0
def beta_exact(a): return np.exp(-mu1 * a) * (1 / (1 - a/2))
def beta_approx(a): return np.exp(-mu1 * a) * (1 + a/2 + a**2/4 + a**3/8)
a1 = np.linspace(0, 0.3, 500); a2 = np.linspace(0, 1.999, 500)
plt.figure(figsize=(14, 6))
plt.subplot(1, 2, 1); plt.plot(a1, beta_exact(a1)); plt.plot(a1, beta_approx(a1), '--'); plt.grid(True)
plt.subplot(1, 2, 2); plt.plot(a2, beta_exact(a2)); plt.plot(a2, beta_approx(a2), '--'); plt.grid(True)
plt.show()
\end{lstlisting}

\subsection*{A.4 Example 3: Modified Bessel Function Birth Law}

\begin{algorithm}[H]
\caption{Polynomial Approximation of Bessel Function Birth Law}
\label{alg:example3}
\begin{algorithmic}[1]
\Require Mortality parameter $\mu_1$, maximum truncation orders $N_{\max}$, age range $[0, a_{\max}]$
\Ensure Convergence comparison of polynomial approximations

\State \textbf{Define exact function:} $\beta_{\text{exact}}(a) \gets e^{-\mu_1 a} \cdot I_0(2\sqrt{a})$
\Statex \Comment{$I_0$ is the modified Bessel function of first kind, order zero}

\Function{BetaApprox}{$a$, $\mu_1$, $N$}
    \State $S \gets 0$
    \For{$n = 0$ \textbf{to} $N$}
        \State $S \gets S + a^n / (n!)^2$
    \EndFor
    \State \Return $e^{-\mu_1 a} \cdot S$
\EndFunction

\State Create age grid $a_i$ for $i = 0, \ldots, M-1$
\State Evaluate $\beta_{\text{exact}}(a_i)$ using library Bessel function
\State Evaluate $\beta_{\text{approx}}(a_i, N)$ for $N = 2, 5, 10, \ldots$
\State Plot all curves to demonstrate series convergence

\State \Return Convergence comparison figure
\end{algorithmic}
\end{algorithm}

\begin{lstlisting}[caption={Python implementation for Example 3}, label=lst:code3]
import numpy as np
import matplotlib.pyplot as plt
from scipy.special import iv
import math
def beta(a, mu1): return np.exp(-mu1 * a) * iv(0, 2 * np.sqrt(a))
def beta_approx(a, mu1, N_max):
    series_sum = np.zeros_like(a)
    for n in range(N_max + 1): series_sum += a**n / (math.factorial(n)**2)
    return np.exp(-mu1 * a) * series_sum
mu1 = 1.0; a_vals = np.linspace(0, 10, 500)
plt.figure(figsize=(8, 6))
plt.plot(a_vals, beta(a_vals, mu1), 'b-', label='Exact')
plt.plot(a_vals, beta_approx(a_vals, mu1, 2), 'r--', label='N=2')
plt.plot(a_vals, beta_approx(a_vals, mu1, 5), 'g:', label='N=5')
plt.legend(); plt.grid(True); plt.show()
\end{lstlisting}

\end{document}